\documentclass[aps,pra,twocolumn,showpacs,superscriptaddress, nofootinbib]{revtex4-2} 
\usepackage{amsmath, amsthm, amssymb,mathrsfs}
\usepackage{graphicx, wrapfig} 

\usepackage{physics}
\usepackage{bbold}
\usepackage{amsmath, amsthm, amssymb,mathrsfs}

\newtheorem{theorem}{Theorem}

\newtheorem{lemma}{Lemma}
\usepackage{mathtools}

\usepackage[normalem]{ulem}
\usepackage{optidef}
\usepackage{algorithmic}
\usepackage{algorithm}

\usepackage{dsfont}
\usepackage{comment}
\usepackage{xcolor}
\usepackage{vwcol}
\usepackage[colorlinks]{hyperref}

\definecolor{MyBlue}{rgb}{0,.2,.8}
\definecolor{MyRed}{rgb}{.8,.2,0}
\hypersetup{
    citecolor=MyBlue,
    colorlinks=true,
    linkcolor=MyRed,
    urlcolor=MyBlue}

\usepackage[acronym]{glossaries}
\makeglossaries
\newacronym{opi}{OPI}{Output permutation invariant}
\newacronym{ejm}{EJM}{Elegant joint measurement}
\newacronym{nsi}{NSI}{No-signalling and independent}
\newacronym{qkd}{QKD}{quantum key distribution}
\newacronym{di}{DI}{device-independent}
\newacronym{pm}{PM}{prepare-and-measure}
\newacronym{sdp}{SDP}{semidefinite programming}
\newacronym{POVM}{POVM}{Positive Operator Valued Measure}
\newacronym{usd}{USD}{Unambiguous state discrimination}
\newacronym{qber}{QBER}{Quantum Bit Error Rate}

\newcommand{\corr}[1]{{\color{black}#1}}

\begin{document}
\author{Antoine Girardin}
\affiliation{Department of Applied Physics University of Geneva, 1211 Geneva, Switzerland}

\author{Nicolas Gisin}
\affiliation{Department of Applied Physics University of Geneva, 1211 Geneva, Switzerland}
\affiliation{Constructor University, Geneva, Switzerland}
\title{Violation of the Finner inequality in the four-output triangle network}

\begin{abstract}
Network nonlocality allows one to demonstrate nonclassicality in networks with fixed joint measurements, that is without random measurement settings. The simplest network in a loop, the triangle, with 4 outputs per party is especially intriguing. The “elegant distribution” [\href{https://doi.org/10.3390/e21030325}{N. Gisin, Entropy 21, 325 (2019)}] still resists analytic proofs, despite its many symmetries. In particular, this distribution is invariant under any output permutation. The Finner inequality, which holds for all local and quantum distributions, has been conjectured to be also valid for all no-signalling distributions with independent sources (NSI distributions). Here we provide evidence that this conjecture is false by constructing a 4-output network box that violates the Finner inequality and prove that it satisfies all NSI inflations up to the enneagon. As a first step toward the proof of the nonlocality of the elegant distribution, we prove the nonlocality of the distributions that saturates the Finner inequality by using geometrical arguments.
\end{abstract}

\maketitle
\section{Introduction}

The study of correlations in networks with independent sources has attracted a lot of attention recently, notably because of its ability to provide nonlocality without input \cite{review, Tavakoli_Review}. 
\corr{
The first studies of nonlocality started in the well-known Bell scenario~\cite{Bell}. In this scenario, two parties share some resources and the goal consists in maximizing a score. More generally, in Bell scenarios n parties share some common resources. The parties can agree on a strategy before the start of the game, but can no longer communicate once the game starts. This game is used to prove that quantum mechanics cannot be explained with local variables. In this scenario, building Bell inequalities allows one to distinguish between local, and nonlocal correlations.
}

Nonlocality in networks differs fundamentally from the standard Bell nonlocality, \corr{in that, some resources are only shared by a subset of the parties}. In fact, Bell inequalities do not allow one to characterize nonlocality in networks, since the local regions in networks are nonconvex~\cite{branciard2012bilocal}. The development of novel methods to study network nonlocality is therefore needed. 

The triangle network has been of particular interest due to its minimal shape. The first example of triangle nonlocality comes from T.Fritz \cite{Fritz_2012}, though it uses the standard CHSH test~\cite{CHSH}. More recently, some distributions called “token counting” have been proved nonlocal \cite{ Renou_2019, Renou_TC_CM_22, Pozas_2023, krivachy_singlephoton}. All these examples of nonlocal distribution in the triangle are for the four-output case, but some distributions with fewer outputs have been found \cite{minimal_triangle_2022}. The 2-output triangle has been studied in detail, no sign of quantum nonlocality has been found yet, and the regions where one could still hope to find some gets smaller and smaller~\cite{gisin2020constraints, pozas_2triangle}. 

In this work, we focus on distributions in the 4-output triangle network without inputs. We study the subspace with distributions invariant under exchange of parties and outputs, that we call output permutation invariant (OPI) for short. This subspace contains the Elegant distribution, introduced in \cite{Gisin_2019}, obtained by using the elegant joint measurement (EJM) on shared maximally entangled two qubits states. This distribution is thought to be nonlocal \cite{krivachy_neural_2020}, but a proof is still awaited.

\begin{figure}[t]
	\centering
	\includegraphics[width=\linewidth]{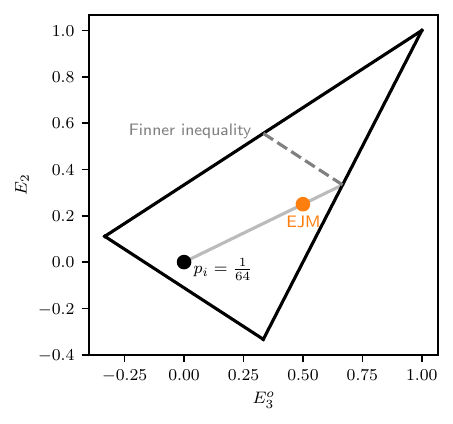}
	\caption{Scheme of the output permutation invariant (OPI) subspace of probability distribution with the Finner inequality and the elegant distribution (called EJM). Here, this subspace is parametrized using the two party marginal and the three-party marginal in a loop, respectively $E_2$ and $E_3^o$, defined in Eq.~\eqref{def_E2E3}. All distributions in the 4-output OPI triangle must be in the above triangle because of the positivity of all probabilities.}
	\label{fig:sheme}
\end{figure}

To characterize this symmetric subspace, we use the inflation technique~\cite{wolfe2019inflation} to exclude regions of the symmetric subspace that do not respect the no-signalling principle with the independence of the sources (NSI). As a first step in the direction of a proof of the nonlocality of the Elegant distribution, we give a proof of the nonlocality of the distributions in this subspace that saturates the Finner inequality, see Fig.~\ref{fig:sheme}. This proof uses a novel method that relies on geometric arguments that may also be useful for other distributions in the triangle network.

\section{Problem and Numerical methods}\label{problemandmethod}

The OPI subspace contains all the probability distributions that satisfy the invariance under exchange of parties, i.e. $p_{abc}=p_{bca}=…$, and the symmetry between all outputs, i.e. $p_{000}=p_{111}=…$ This subspace is two-dimensional for the 4-output triangle. Indeed, the only three different possible probabilities are $p_{111}$, when all parties give the same output, $p_{112}$, when two parties give the same output but not the third, and $p_{123}$, when all parties give a different output. These three probabilities have the additional constraint to sum to one:


\begin{equation}\label{psum1}
    4 p_{111} + 36 p_{112} + 24 p_{123} = 1
\end{equation}

We can also parametrize this subspace with correlators. We write the four outputs of each party with two bits $a= (a_0, a_1) \in \{-1,1\}^2$. It is convenient to define an additional bit $a_2 =a_0\cdot a_1$. Then, the only two nonvanishing correlators are the two party marginal $E_2$, \corr{later referred to as the "two-party correlator",} and the three-party marginal in a loop $E_3^o$ defined as

\begin{equation}\label{def_E2E3}
\begin{split}
        E_2 &= \langle a_j \cdot b_j\rangle = \langle a_j \cdot c_j\rangle = \langle b_j \cdot c_j\rangle \\
     E_3^o &= \langle a_j \cdot b_k \cdot c_l\rangle \\
     \text{where     } j, k, l &\in\{0,1,2\}, j\neq k\neq l \neq j
\end{split}
 \end{equation}

A linear transformation allows one to go from one parametrization to the other with the relation \eqref{corr_to_probs}.

\begin{equation}\label{corr_to_probs}
\begin{pmatrix}
	1 \\ E_2 \\ E_3^o
\end{pmatrix}=
\begin{pmatrix}
		4 & 36 & 24\\
		4 & 4 & -8\\
		4 & -12 & 8\\
\end{pmatrix} \cdot
\begin{pmatrix}
p_{111}\\ p_{112}\\ p_{123}
\end{pmatrix}
\end{equation}

The Finner inequality \cite{finner1992generalization} implies that 
\begin{equation}\label{Finner}
    p_{abc}\leq \sqrt{p(a)p(b)p(c)} 
\end{equation}

for any local or quantum distributions~\cite{renou2019limits}, with $p(a)$, $p(b)$, $p(c)$ the one-party marginals. This inequality is illustrated in the Fig.\ref{fig:sheme} together with the positivity constraints that form a triangle, the distribution obtained with the EJM and the fully noisy distribution. Note that here the nonlinear Finner inequality~\eqref{Finner} appears linear because all marginals $p(a)$, $p(b)$, $p(c)$ are set to $\frac{1}{4}$ by the OPI condition.

To bound the NSI region, we consider polygon inflations to find an upper bound on the two-party correlator $E_2$. Increasing the number of parties allows one to get more constraints on this correlator, because it appears in all polygons. One should notice that one cannot constrain the other correlator of the triangle, $E_3^o$, the tripartite correlator in a loop, since bigger polygons don't contain it. \corr{We introduce polygon inflations with more details in the appendix~\ref{appendix:inflation}.}

The constraints come from the NSI condition, \corr{as in Ref.~\cite{gisin2020constraints}}. If Alice locally modifies the topology of the network, the statistics should not be modified for Bob and Charlie, otherwise Alice could signal to Bob and Charlie.
\corr{For instance, in the first level of the inflation, inflating the topology from a triangle to a square, the nontrivial correlator appears 

\begin{equation}\label{example_square}
\begin{split}
        E &= \langle a_j \cdot c_j\rangle = \langle b_j \cdot d_j\rangle  \\
        &= \langle a_j \rangle \cdot \langle c_j \rangle \\
        &=0
\end{split}
 \end{equation}

 This correlator $E$ is the two-party correlator for nonconnected parties in the square network. The independence of the sources allows one to separate this two-party correlator by the product of the two one-party marginals in Eq.~\ref{example_square}. The symmetries of the problem implies that the one-party marginal is null, leading to the constraint that this correlator is null. The NSI condition allows one to conclude that the correlator $E_2$ in the square network has to be the same as the one in the triangle network. This allows the use of the inflated network to constraint correlators in the triangle. With the relation between the correlators and the probabilities, see Eq.~\ref{corr_to_probs} for the triangle and Eq.~\ref{matrix_square} for the square network, this leads to constraints on the probabilities too. We build all the constraints in this way and list them in the appendix~\ref{details_on_correlators}.
}

Our approach is not completely general because we suppose that all the sources are identical (but independent). In principle, the sources could distribute different correlations, which may lead to OPI distribution unachievable with identical sources. This additional assumption simplifies the problem significantly. Without this assumption, one could still use inflations with a number of parties that are a multiple of 3 (as in Ref.~\cite{gisin2020constraints}), so that each source appears the same number of times in the inflated network. Here, this constraint leads to a stricter bound on $E_2$ compared to the general NSI condition.

For the inflation, we use two different numerical methods. 
The first method exploits the Gurobi optimizer~\cite{gurobi} that allows one to optimize an objective with linear and quadratic constraints. We set the permitted violation of the constraint to the smallest possible value, $10^{-9}$, in order to recover more precise results.
The second method linearizes the quadratic constraints. This allows one to significantly speed up the optimization. For this method, we replace the correlator $E_2^2$ by $\bar{E_2}^2+\epsilon$, with $\bar{E_2}$ a constant that approximates the maximal value for $E_2$, our target. The parameter $\epsilon$ becomes the new parameter to maximize over.
The linearization of the quadratic constraint uses the approximation for a small $\epsilon<<\bar{E_2}^2$.

\begin{equation}\label{eq:num_approx}
\begin{split}
    E_2 = \sqrt{\bar{E_2}^2+\epsilon} 
     &\approx \bar{E_2} \left(1+\frac{\epsilon}{2\bar{E_2}^2}-\frac{\epsilon^2}{8\bar{E_2}^4}\right) \\
     &\leq \bar{E_2} \left(1+\frac{\epsilon}{2\bar{E_2}^2}\right) \\
\end{split}
\end{equation}

This approximation gets better by recursively maximizing $\epsilon$ and updating our value for $\bar{E_2}$. When $\epsilon$ gets comparable with the numerical imprecision, $\epsilon\approx 10^{-10}$, the error due to the approximation becomes negligible, and the bound found for $E_2$ is very reliable. The same method is used for other nonlinear constraints, such as correlators equal to $E_2\cdot E_3$ in the heptagon, 
\corr{where $E_3$ is the three-party correlator in a line}
(see appendix~\ref{details_on_correlators}).
As shown in Eq.~\eqref{eq:num_approx}, the value one converges to is ideally slightly greater than the exact bound for $E_2$. This is better than converging to a smaller value, since we are looking for an upper bound on $E_2$. It is the reason we add this $\epsilon$ to $\bar{E_2}^2$ instead of $E_2$, since the approximation would be smaller than the exact bound on $E_2$.
This second method is much faster than the quadratic solver Gurobi. As a trade-off, it does not give a solution containing exact zeros, which is relevant in the section~\ref{analytical_computation}. Moreover, Gurobi is able to give the number of optimal solutions by scanning exhaustively the parameters linked with the nonlinear constraints.

\begin{center}
	\begin{table}
		\centering
		\begin{tabular}{|c|c|c|}
			\hline
			Nb vertices & $E_2$ max  & $E_2$ max with smaller polygons \\
			\hline
			3 & 1 & 1\\
			\hline
			4 & 0.5 & 0.5\\
			\hline
			5 & 5/11 &  5/11 \\
			\hline
			6 & $\sqrt{2}-1$ & 0.404040\\
			\hline
			7 & 0.393141 & 0.392034 (0.392037)\\
			\hline
			8 &  0.381966 (0.381966) & 0.38003 (0.379197)\\
			\hline
			9 &  N/A (0.376608) &0.37491 (0.375051)\\
			\hline
		\end{tabular}
		\caption{Results of inflations with both numerical methods for polygons with up to nine edges. When the second method gives a different upper bound than Gurobi, its value is given in brackets. The first column considers only constraints from the current polygon, while the column “$E_2$ max with smaller polygons” also considers constraints due to all the smaller polygons. \corr{We could not find the result for 9 vertices without the smaller polygons with Gurobi in a respectable amount of time, see end of appendix~\ref{details_on_correlators}.}}
		\label{tab:results}
	\end{table}
\end{center}

\section{NSI bound on $E_2$}
\subsection{Numerical inflation}

We consider that the sources distribute some correlations to the parties and that all the sources distribute the same correlations. 
The inflations up to five vertices lead to linear constraints that can be solved with linear programming. The inflations with at least six vertices contain quadratic constraints, because two independent $E_2$ are allowed in the hexagon, see the appendix~\ref{details_on_correlators} for more details on correlators.

The maximal values of the correlator $E_2$ given by both numerical methods for any inflation up to nine are given in Table~\ref{tab:results}. In general, we can keep the constraints from the smaller polygons, but we add the results obtained by using the constraints from the largest polygon only. This could help to find an analytical structure to the results, as initiated in the section~\ref{analytical_computation}. We couldn't get a value with Gurobi for the last polygon without the smaller ones because the optimization was too slow.

\begin{figure}[t]
	\centering
	\includegraphics[width=\linewidth]{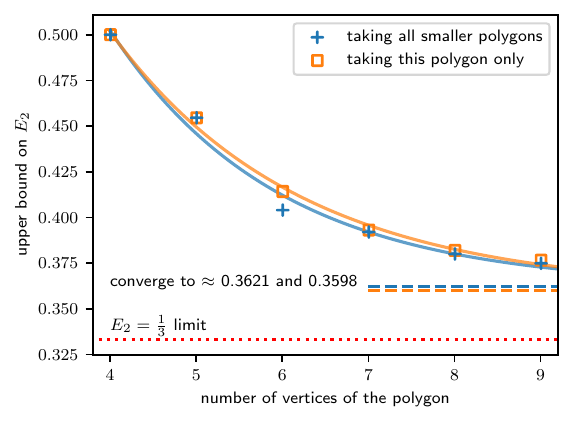}
	\caption{Results of the inflation for different sizes of the polygon. An exponential fit is added with the value it converges to.}
	\label{fig:plotinflation}
\end{figure}

\begin{figure*}[t]
	\begin{minipage}{0.48\textwidth}
		\centering
		\includegraphics[width=1\linewidth]{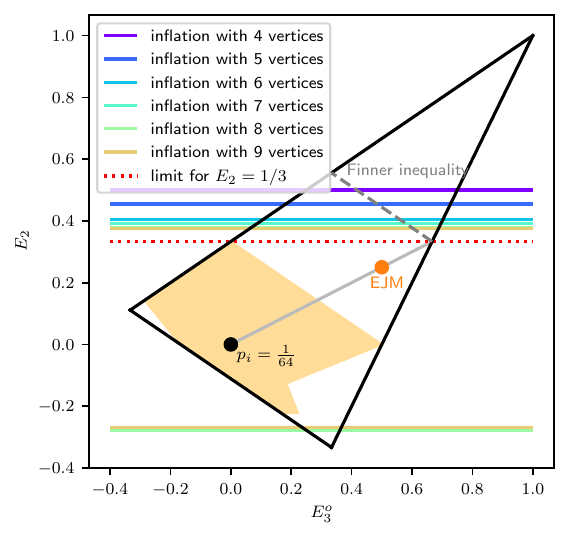}
	\end{minipage}\hfill
	\begin{minipage}{0.48\textwidth}
		\centering
		\includegraphics[width=1\linewidth]{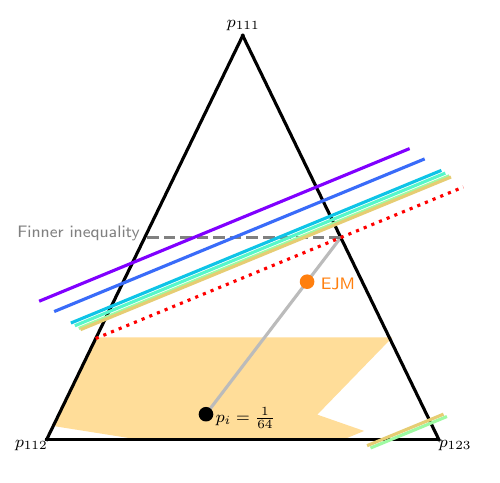}
	\end{minipage}
 \caption{Space of probability distributions with $E_2$ and $E_3^o$ (left) and the probabilities $p_{111}$, $p_{112}$ and $p_{123}$ (right). The known local region~\cite{local_region} is coloured in light orange. The EJM point is the distribution obtained with the Elegant Joint Measurement. The lines referenced in the legend are the bound on $E_2$ obtained with different levels of the inflation.}
 \label{fig:subspace_inflation}
\end{figure*}

Note that the upper bound on $E_2$ has to stay above $E_2=\frac{1}{3}$, since a local model reaching this value has been found~\cite{local_region} in this OPI subspace. Interestingly, $E_2=\frac{1}{3}$ is also obtained at the special point, where the Finner inequality meets the positivity constraint, as well as the straight line going through the fully noisy distribution and the elegant distribution.

Surprisingly, the upper bound imposed by the NSI condition seems to converge to a value greater than $\frac{1}{3}$. This suggests that NSI distributions can violate the Finner inequality, which goes against the conjecture proposed in Ref.~\cite{renou2019limits}, stating that the Finner inequality captures the limit of correlations possible in any NSI theory. \corr{To prove that our conjecture is true, one would need to show that no polygon inflation allows one to violate the Finner inequality, which could only be done analytically. We give the first analytical result in the section~\ref{analytical_computation}}.

We could not give a tight upper bound for $E_2$, but we can certify that it is between $0.37491$ and $\frac{1}{3}$. An exponential fit plotted in the Fig.~\ref{fig:plotinflation} converges to $\sim 0.36$. The only motivation for this exponential fit is empirical. In principle, the upper bound on $E_2$ could follow any decreasing function greater than $\frac{1}{3}$, and it is possible that this upper bound does not follow any analytical function. We additionally plot in the Fig.~\ref{fig:subspace_inflation} the symmetric subspace with the bounds given by each order of the inflation. 

We can use the same method to minimize $E_2$ and find a lower bound. The inflation of the order 8 and 9 gives us a nontrivial bound plotted in the Fig.~\ref{fig:subspace_inflation}. As for the upper bound, we could not find an optimal value for the lower bound. The lowest $E_2$ obtained with a known local model is $-\frac{2}{9} = -0.\bar{2}$~\cite{local_region}, which is not very far from our lower bound at $E_2=-0.2690928$.

\subsection{Analytical computation of a bound on a correlator}\label{analytical_computation}
\corr{
In order to find an exact bound for the correlator $E_2$, one needs to find a structure in order to prove analytically a convergence when the number of vertices becomes arbitrarily large. With this objective, we give a method that allows one to give an exact value for the hexagonal inflation, the first order unsolvable with linear programming because of the quadratic constraints.

Indeed, 
}
for the hexagon, one can analytically prove with the output of Gurobi that the value is exactly $\sqrt{2}-1$ in the following way: some probabilities $p_j$ in this hexagonal network that maximizes $E_2$ are $0$. These probabilities $p_j$ are the probabilities of the 33 different outputs given by all parties. We have for instance the probability that all parties give the same output, and all the other OPI outcomes.

One has $p_j=0 \implies \sum_k C_{jk} E_k =0$ with $C_{jk}$ the matrix that relates the probabilities to the correlators. For each of them, one has one parameter $q_j$ that we can vary to cancel nonzero correlators. At the end, this leads to the equation

\begin{center}
\begin{equation}
	\begin{split}
	&\sum_j \sum_k q_j C_{jk} E_k =0 \\
\text{const.} + x E_2 + & y E_2^2 + 0 \cdot(\text{other correlators}) =0
	\end{split}
\end{equation}
\end{center}

The constant comes from the normalization that has to be added to relate probabilities to correlators. The correlators $E_2$ and $E_2^2$ are the only two correlators one doesn't need to cancel, since they only contain the correlator $E_2$ that appears in the triangle. This leads to an exact number for the upper bound on $E_2$.

This method works for the hexagon, we have 22 parameters $q_j$ and 20 correlators to cancel. It leads to the equation 

\begin{equation}
	\begin{split}
		\frac{1}{256} + E_2^2 &\cdot (-\frac{1}{256})+E_2 \cdot (-\frac{1}{128})=0\\
		&\implies E_2=\sqrt{2}-1
	\end{split}
\end{equation} 

This method does not allow one to find an exact value for the next order of the inflation. Indeed, the solution given by the quadratic solver does not contain enough zeros to generate a nontrivial null space of the matrix $C_{jk}$.

\section{Result on the nonlocality of the distributions saturating the Finner inequality}
Proving the nonlocality of a given distribution is a difficult problem. This can be done in some cases with an inflation of the network or, if the distribution is token-counting, it is possible to prove its nonlocality in some cases.

A local distribution is a distribution that can be obtained using classical resources. It is sufficient to consider that each source distributes a number of symbols to the connected parties \cite{Rosset2018}. We can then map this problem to a 3-dimensional cube, where each axis represents one source~\cite{renou2019limits}. Each point $(\alpha, \beta, \gamma)\in [0,1]^3$ in the cube corresponds to the case when the sources have distributed the value $(\alpha, \beta, \gamma)$ and the three parties Alice, Bob, and Charlie have outputted $a=s_A(\beta, \gamma), b=s_B(\alpha, \gamma), c=s_C(\alpha, \beta)$, with $s_A, s_B, s_C$ the local strategies of the three parties. The probability $p_{abc}$ of Alice outputting $a$, Bob $b$, and Charlie $c$ correspond now to a volume in this cube.

\corr{
Any local distribution can then be constructed in this cube. Let's choose the colours white, blue, red, and green for the four outputs $0$, $1$, $2$, $3$. The local strategy of each party will then be a coloured square, illustrating the output of the party given the received symbols. These squares are the faces of the cube that uniquely define the final distribution. Alternatively, one can start building the distribution from different probabilities $p_{abc}$ represented as volumes in the cube and deduce the local strategy of each party ultimately. We use this second approach for our proof. 
}

We present now a new method to prove the nonlocality of a distribution using geometrical arguments in this cube. This method allows one to have a simple and understandable proof that does not rely on complicated inequality found by a computer, as we get with inflations. The proof is illustrated in the appendix~\ref{appendix:proof}.
We first introduce the lemma~\ref{lemma} that we use in the proof of the theorem~\ref{thm}.

\begin{lemma}\label{lemma}
	If $p_{000}=\frac{1}{8}$ and the marginals are $p(a=0)=p(b=0)=p(c=0)=\frac{1}{4}$, then $p_{000}$ is equivalent to a cube in the cube representation.
\end{lemma}

\begin{proof}
The Finner inequality implies $p_{abc}\leq \sqrt{p(a)p(b)p(c)}$. The values we have in the lemma imply that we saturate the Finner inequality, so the shape of $p_{000}$ is a rectangular parallelepiped in the cube. The rectangular parallelepiped has length $x$, $y$ and $z$ and because the marginals are the same, we have $xz=yz=xy$. The first equality implies $x=y$, the last implies $x=z$, so the proof is complete.
\end{proof}
 
\begin{theorem}\label{thm}
No local distribution can saturate the Finner inequality and be OPI in the 4-output triangle network.
\end{theorem}

\begin{proof}
(by contradiction) 

Let's take a probability distribution and suppose it saturates the Finner inequality and is invariant under exchange of parties and output. We now build the most general local strategy that achieves this probability distribution.

The distributions that saturate the Finner inequality have $p_{000}=\sqrt{p(a=0)p(b=0)p(c=0)}$. Since $p(a=0)=p(b=0)=p(c=0)$ and $p(a=0)=p(a=1)=p(a=2)=p(a=3)=\frac{1}{4}$ we have $p_{000}=p_{111}=p_{222}=p_{333}=\frac{1}{8}.$

Let's choose the output '0'. We can choose to order the labels to start in each axis with the column that contains the most of output '0', \corr{here it means that we place the volume corresponding to $p_{000}$ in a corner}. We know that the volume of this region corresponding to the output $p_{000}$ is $\frac{1}{8}$. (Fig.\ref{Fig:1})

Because the strategy is invariant under exchange of parties and outputs, we necessarily have $p(a=x)=p(b=x)=p(c=x)= \frac{1}{4}$ for $x$ any outputs because  $p_{xxx}=\frac{1}{8}$ for all $x$, so the only possible shape for $p_{000}$ is a cube as implied by the Lemma~\ref{lemma}. (Fig.\ref{Fig:2})

For the same reason, the volume corresponding to the other $p_{xxx}$ will necessarily have a shape equivalent to a cube (meaning there exists a reorder of the symbols that leads to a cube). 

In fact, \corr{$p_{111}$, $p_{222}$, and $p_{333}$} will be three cubes because $p_{000}$ being a cube imposes for instance that A, B, and C output '0' if they receive $\alpha, \beta, \gamma \leq \frac{1}{2}$, with $\alpha, \beta, \gamma$ the shared randomness. So for $p_{111}$, at least two of the \corr{parameters} $\alpha, \beta, \gamma$ should be greater than $\frac{1}{2}$. This can only be achieved by a shape equivalent to a cube taking place for $\alpha, \gamma \geq \frac{1}{2}$ or any other pair of $\alpha, \beta, \gamma$. This shape equivalent to a cube can only be a sliced cube. \corr{This cube can be sliced in the direction of the} $\beta$ axis (or the last direction we did not pick in the pair of $\alpha, \beta, \gamma$ previously) (Fig.\ref{Fig:3}). \corr{Finally}, if $p_{111}$ is not a cube for exactly  $\alpha, \gamma \geq \frac{1}{2}$ and $\beta \leq \frac{1}{2}$, there is not enough room left for $p_{222}=\frac{1}{8}$. Indeed, the maximum volume of $p_{222}$ is given by $\frac{1}{2}\times\frac{1}{2}\times(\frac{1}{2}-\delta)$, with $\delta$ the total length of $\beta \geq \frac{1}{2}$. The only solution for $p_{222}=\frac{1}{8}$ is $\delta=0$, meaning that $p_{111}$ has the shape of a cube.

So the only possibility for $p_{xxx}$ are cubes (Fig.\ref{Fig:4}), which does not lead to a distribution invariant under the symmetries we suppose (some $p_{123}=0$, but four of them are $1/8$), and this contradiction ends the proof.
\end{proof}

\section{Conclusion}
We have shown with the inflation technique that a large region of the output permutation invariant (OPI) subspace is not no-signalling with independent sources (NSI). Interestingly, our method seems to leave a NSI region above the Finner inequality. This makes us conjecture that NSI correlations exist beyond the Finner inequality.

More specifically, for the OPI subspace, we couldn't find a tight upper bound of the two-party marginal $E_2$. This would require proving a structure for every level of the inflation and computing where it converges. We could only find an analytical expression in the hexagon and the smaller polygons. Therefore, a novel idea is needed to find an exact bound for $E_2$. 

As a first step to prove the nonlocality of the elegant distribution, we gave an analytical proof of the nonlocality of the distributions that saturate the Finner inequality on the OPI subspace. To prove this, we used a novel idea using geometric arguments.

A proof of the nonlocality of the elegant distribution is naturally still a crucial direction for future research. Using a similar idea to the proof we gave may be helpful. It would require abandoning the properties of the distribution that saturates Finner, leading in general to many more local models to rule out.

\medskip

\emph{Acknowledgements.---} We thank Marc-Olivier Renou, Sadra Boreiri, Tamás Kriváchy, Alejandro Pozas-Kerstjens, and Victor Gitton for discussions and comments. We thank Bernard Gisin for the second optimization method. We acknowledge financial support from the Swiss National Science Foundation (project 2000021\_192244/1 and NCCR SwissMAP).

\section{Code availability}
We provide the code to realize the inflation with Gurobi at the link \url{https://github.com/Antoine0Girardin/Inflation-OPI-4-triangle}

\onecolumngrid
\appendix

\section{Polygon inflations}
\label{appendix:inflation}

The inflation technique consists in inflating the network in order to constrain the correlations in the original network. Depending on the type of constraints one adds on the inflated network, it is possible to constrain local, quantum or NSI correlations.

We consider a polygon inflation. To the best of our knowledge, this is the only useful NSI inflation for the triangle network. By supposing that all sources are equal, it is possible to consider every polygon, starting from the triangle, see Fig.~\ref{fig:scheme_inflation}, where the vertices represent the parties, and the edges show where the sources distribute correlations. The first level of the inflation has the shape of a square, with an additional source and a fourth party. The procedure can be continued for an arbitrarily large number of sources and parties.

For each level of the inflation, new constraints can be added using the no-signalling condition, see appendix~\ref{details_on_correlators}. The two-party correlator $E_2$ that appears in the original triangle network as well as in every inflated network can then be constrained. This polygon inflation technique does not allow one to constrain the other correlator of the triangle, the three-party correlator in a loop $E_3^o$, since no such loop exists in the other polygons.

\begin{figure*}[h]
		\centering
		\includegraphics[width=1\linewidth]{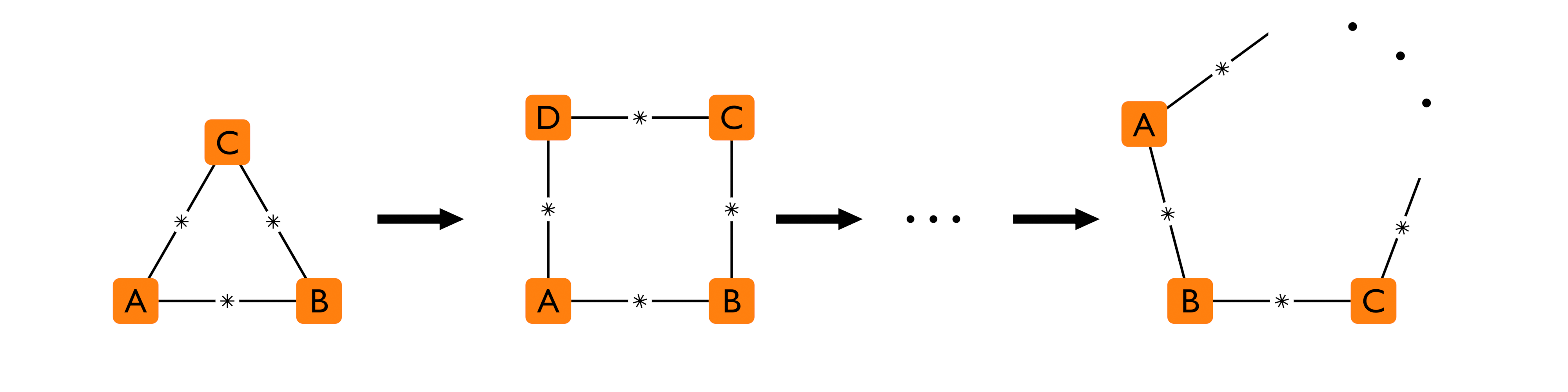}
 \caption{(left) The original triangle network. (middle) The first level inflation: the square inflation. (right) The general polygon inflation for an arbitrary number of parties.}
 \label{fig:scheme_inflation}
\end{figure*}

\section{Technical details about correlators} \label{details_on_correlators}
As mentioned in the main text, we write the four outputs of each party with two bits $a= (a_0, a_1) \in \{-1,1\}^2$, and define an additional bit $a_2 =a_0\cdot a_1$.
In general, we denote the correlators as $abcde…$, with each letter the label $j,k,l \in \{0,1,2\}$ of the correlated bit for the corresponding party. If the party does not have a correlated bit, we write a $0$ instead. The correlator $E_2 = \langle a_j \cdot b_j \rangle$ in the triangle will be written as $jj0$. In the square, the three-party correlator in a line $\langle a_j\cdot b_k\cdot d_j\cdot d_k\rangle = \langle a_j\cdot b_k\cdot d_l\rangle = \langle a_j\cdot b_k\cdot c_l\rangle$ will be noted $jkl0$.
Because of the symmetries, we have $jj0=0jj=j0j=kk0=…$

In the triangle, only two correlators are not trivially zero: $E_2 = jj0$ and $E_3^o = jkl$. In this network with these symmetries, we only have three different probabilities: the probability $p_{111}$ that all outputs are the same, the probability $p_{112}$ that two outputs are the same and the last is different, and finally $p_{123}$, the probability that all outputs are different.
We can relate the probabilities to the correlators with the equation \eqref{corr_to_probs}.

For the square, we have six different correlators:
\begin{tabular}{ c c c c c c c c c c } 
$jjjj$ & $jjkk$ & $jj00$ & $jkjk$ & $jkl0$ & $j0j0$ & 
\end{tabular}
and only one linear constraint $j0j0 = 0$, because $j0j0 = j00 \times j00$ and $j00=E_1=0$ is imposed by the symmetries. The matrix $C$ that links the correlators to the probabilities is given in the equation \eqref{matrix_square} with the first line being the normalization.

\begin{equation}\label{matrix_square}
	C=
	\begin{pmatrix}
		4& 48& 24& 96& 12& 48& 24& \\
		4& -16& 24& -32& 12& -16& 24& \\
		4& -16& 8& 0& -4& 16& -8& \\
		4& 16& 8& 0& -4& -16& -8& \\
		4& -16& -8& 32& 12& -16& -8& \\
		4& 0& -8& 0& -4& 0& 8& \\
		4& 16& -8& -32& 12& 16& -8& \\
	\end{pmatrix}
\end{equation}

For the pentagon, we have ten different correlators : 

\begin{tabular}{ c c c c c c c c c c c c c c c } 
$jjjj0$ & $jjjkl$ & $jjkjl$ & $jjkk0$ & $jjk0k$ & $jj000$ & $jkjk0$ & $jkl00$ & $jk0l0$ & $j0j00$ & 
\end{tabular}

The two linear constraints are $jk0l0=0$ and $j0j00=0$. 

For the hexagon, we have 32 correlators : 

\begin{tabular}{ c c c c c c c c c c c c c c c } 
$jjjjjj$ & $jjjjkk$ & $jjjj00$ & $jjjkjk$ & $jjjkl0$ & $jjjk0l$ & $jjj0j0$ & $jjkjjk$ & $jjkjl0$ & $jjkj0l$ & $jjkkll$ & $jjkk00$ & $jjklj0$ & $jjklkl$ & $jjkllk$ \\ 
$jjk0k0$ & $jjk00k$ & $jj0jj0$ & $jj0kk0$ & $jj0000$ & $jkjk00$ & $jkjlj0$ & $jkjlkl$ & $jkj0k0$ & $jkljkl$ & $jkl000$ & $jk0jk0$ & $jk0kj0$ & $jk0l00$ & $j0j000$ \\ 
$j0k0l0$ & $j00j00$ & 
\end{tabular}

We have 10 linear constraints : $jjj0j0=0$, $jjk0k0=0$, $jkj0k0=0$, $jk0jk0=0$, $jk0kj0=0$, $jk0l00=0$, $j0j000=0$, $j0k0l0=0$, $j00j00=0$, $jj0kk0=jj0jj0$.

The quadratic constraint is $jj0jj0=jj0000 \times jj0000$

For the heptagon, we have the 72 correlators : 

\begin{tabular}{ c c c c c c c c c c c c c } 
$jjjjjj0$ & $jjjjjkl$ & $jjjjkjl$ & $jjjjkk0$ & $jjjjk0k$ & $jjjj000$ & $jjjkjjl$ & $jjjkjk0$ & $jjjkj0k$ & $jjjkkj0$ & $jjjkkkl$ & $jjjkklk$ & $jjjkl00$ \\ 
$jjjk0l0$ & $jjjk00l$ & $jjj0j00$ & $jjj0kl0$ & $jjkjjk0$ & $jjkjkj0$ & $jjkjkkl$ & $jjkjklk$ & $jjkjlkk$ & $jjkjl00$ & $jjkj0jk$ & $jjkj0l0$ & $jjkj00l$ \\ 
$jjkkjj0$ & $jjkkjlk$ & $jjkkll0$ & $jjkkl0l$ & $jjkk000$ & $jjklj00$ & $jjklkl0$ & $jjklk0l$ & $jjkllk0$ & $jjkl0j0$ & $jjkl0kl$ & $jjkl0lk$ & $jjk0jl0$ \\ 
$jjk0j0l$ & $jjk0k00$ & $jjk0lj0$ & $jjk00k0$ & $jjk000k$ & $jj0jj00$ & $jj0jkl0$ & $jj0j0j0$ & $jj0kjl0$ & $jj0kk00$ & $jj0k0k0$ & $jj00000$ & $jkjkjkl$ \\ 
$jkjkl0l$ & $jkjk000$ & $jkjlj00$ & $jkjlkl0$ & $jkjlk0l$ & $jkjl0j0$ & $jkj0jl0$ & $jkj0k00$ & $jkljkl0$ & $jklj0j0$ & $jkl0000$ & $jk0jk00$ & $jk0j0k0$ \\ 
$jk0kj00$ & $jk0k0j0$ & $jk0l000$ & $jk00l00$ & $j0j0000$ & $j0k0l00$ & $j00j000$ & 
\end{tabular}

The 27 linear constraints are $jjjk0l0=0$, $jjj0j00=0$, $jjj0kl0=0$, $jjkj0l0=0$, $jjkl0j0=0$, $jjk0jl0=0$, $jjk0j0l=0$, $jjk0k00=0$, $jjk0lj0=0$, $jjk00k0=0$, $jkjl0j0=0$, $jkj0jl0=0$, $jkj0k00=0$, $jklj0j0=0$, $jk0jk00=0$, $jk0j0k0=0$, $jk0kj00=0$, $jk0k0j0=0$, $jk0l000=0$, $jk00l00=0$, $j0j0000=0$, $j0k0l00=0$, $j00j000=0$, $jj0j0j0=0$, $jj0k0k0=0$, $jj0kjl0=jj0jkl0$, $jj0kk00=jj0jj00$.

The two quadratic constraints are $jj0jj00=jj00000\times jj00000$, $jj0jkl0=jj00000\times jkl0000$.

With the same method we have 236 correlators for the octagon, 114 linear constraints and 6 quadratic ones.

Finally, for the enneagon, we have 702 correlators, 395 linear constraints and 14 quadratic ones. For this polygon, we have correlators like $jj0jj0jj0=jj0000000\times jj0000000 \times jj0000000 $, but we can replace this cubic equation with the quadratic one $jj0jj0jj0=jj0000000\times jj0jj0000$ since we already have the constraint $jj0jj0000= jj0000000 \times jj0000000$.

The computational time is too big to continue this inflation for higher polygons. On an Intel Core i7-1185G7, the Gurobi optimization took 7 minutes for the heptagon, but could not finish in 24 hours for the octagon. The solution given in the Table~\ref{tab:results} was found after 5 hours. We could not find a solution \corr{in 168 hours} for the enneagon with Gurobi. 

In the Table~\ref{tab:results}, we have added a column “$E_2$ max with previous polygons”. These results are obtained by adding the constraints of all previous polygons and constraints like $jj0=jj00$, meaning that the two-party correlator should be the same for any polygon. This allows a slightly better upper bound on $E_2$ and simplifies the computations because replacing correlators in the large polygon by the same one in the smaller polygons leads to simpler constraints. \corr{For instance, for the enneagon, the constraints with an $E_2^2$ have 703 probabilities linked with the $E_2$ correlator that need to be squared, leading to many possible branchings during the Gurobi optimization procedure. When replacing this $E_2$ in the enneagon by the $E_2$ in the triangle, the 703 probabilities are replaced by only 3, reducing the number of quadratic terms and the computational time}. For comparison, the heptagon with previous polygons takes 5 seconds, we have a solution for the octagon after a few seconds, and the full optimization takes a few hours.

\section{Illustrations of the proof of the theorem~\ref{thm}}
\label{appendix:proof}
\begin{figure}[!htb]
    \begin{minipage}{0.48\textwidth}
        \centering
        \includegraphics[width=.8\linewidth]{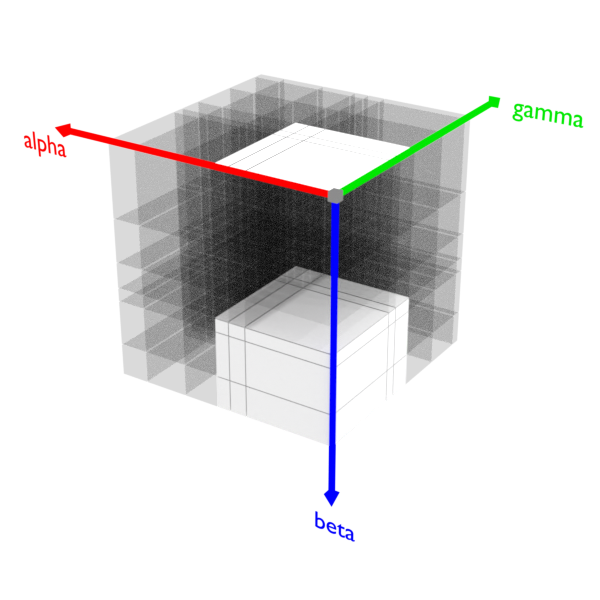}
        \caption{Cube with ordered output 0 (white) in the cube representation. \corr{This white box has a volume that corresponds to the probability $p_{000}$.}}\label{Fig:1}
    \end{minipage}\hfill
    \begin{minipage}{0.48\textwidth}
        \centering
        \includegraphics[width=.8\linewidth]{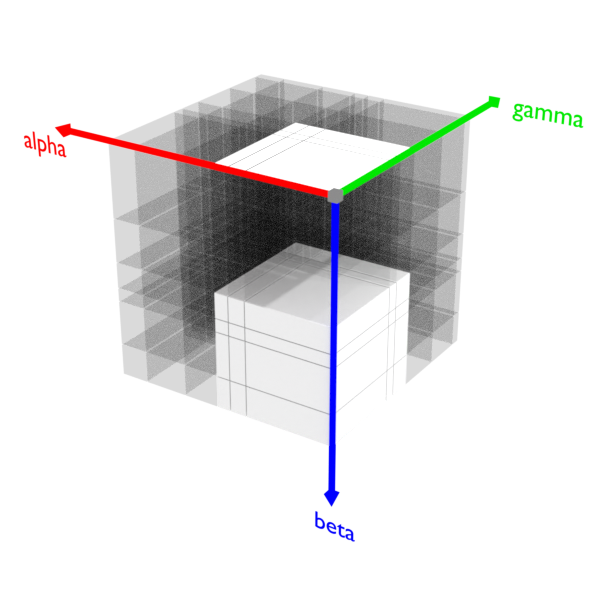}
        \caption{Cube with the only possible size for the output 0 (a cube of side $1/2$).}\label{Fig:2}
    \end{minipage}

\end{figure}

\begin{figure}[h]
\begin{minipage}{0.48\textwidth}
	\centering
	\includegraphics[width=.8\linewidth]{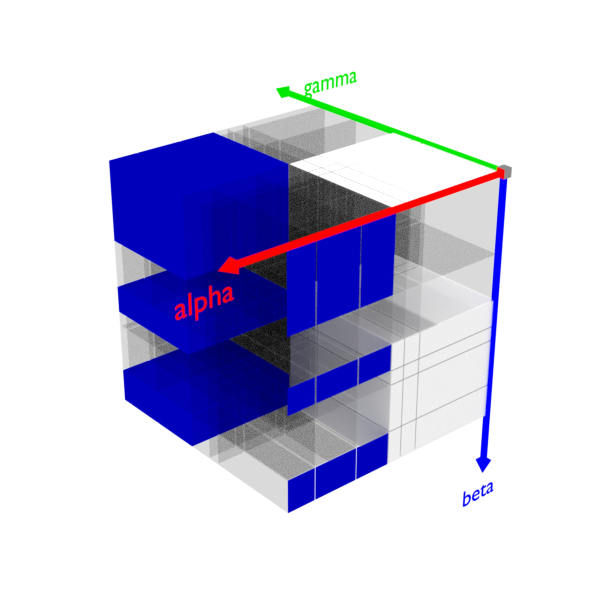}
	\caption{A sliced cube strategy for the output 1 (blue). \corr{This blue volume corresponds to the probability $p_{111}$. Its shape is equivalent to a cube, in the sense that there exists a reordering of the axis $\beta$ that leaves it as a cube.}}\label{Fig:3}
\end{minipage}\hfill
\begin{minipage}{0.48\textwidth}
	\centering
	\includegraphics[width=.8\linewidth]{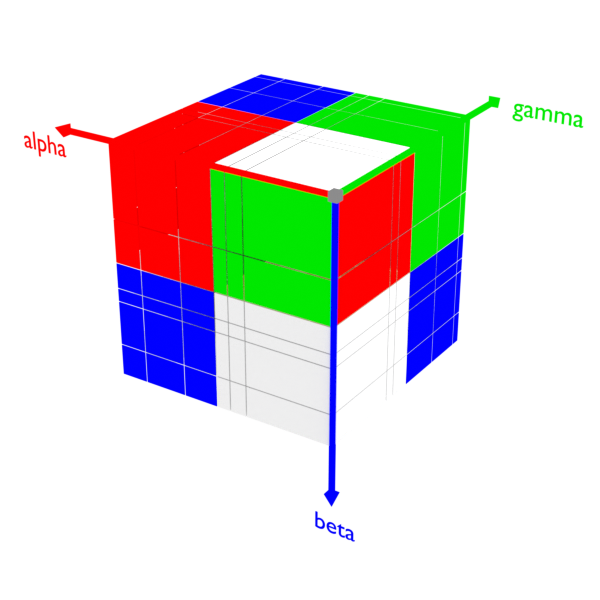}
	\caption{Fully coloured cube that saturates the Finner inequality and satisfies the condition $p_{111}=1/8$, but does not give an OPI distribution.}\label{Fig:4}
\end{minipage}
\end{figure}

\newpage

\twocolumngrid
\bibliography{main}

\end{document}